\documentclass[11pt]{article}

\usepackage[margin=1in]{geometry}
\usepackage{amsmath,amssymb,amsthm,mathtools}
\usepackage{microtype}
\usepackage[hidelinks]{hyperref}
\hypersetup{
  pdftitle={No Welfare Ties Between Pure-Strategy Nash Equilibria},
  pdfauthor={Moeen Nehzati},
  pdfsubject={Research note}
}

\newtheorem{theorem}{Theorem}
\newtheorem{proposition}[theorem]{Proposition}
\newtheorem{corollary}[theorem]{Corollary}
\newtheorem{lemma}[theorem]{Lemma}
\newtheorem{definition}[theorem]{Definition}
\newtheorem{remark}[theorem]{Remark}

\DeclareMathOperator{\NE}{NE}
\DeclareMathOperator{\Conf}{Conf}

\newcommand{\R}{\mathbb{R}}
\newcommand{\U}{\mathcal{U}}
\newcommand{\Pcal}{\mathcal{P}}

\title{No Welfare Ties Between Pure-Strategy Nash Equilibria}
\author{Moeen Nehzati}
\date{October 3, 2026}

\begin{document}

\maketitle

\begin{abstract}
We show that distinct pure-strategy Nash equilibria typically have different utilitarian welfare values in sufficiently rich classes of \(C^2\) games with open strategy sets. Here, typical means both generic and finitely prevalent, while rich means that the class admits a finite-dimensional family of perturbations that can independently move payoff values and first derivatives in any direction at any two distinct strategy profiles—for example, polynomial perturbations of degree at most three. Consequently, a welfare-maximizing equilibrium, if one exists, is unique. The result extends to each fixed \(C^1\) payoff aggregator whose derivative with respect to the payoff vector is everywhere nonzero. In the full \(C^2\) utility space, we also treat scalar criteria depending on strategies, payoffs, and payoff derivatives, subject to a nonvanishing condition in payoff-level or cross-player derivative directions. Finally, adjoining arbitrarily small perturbations from a suitable finite-dimensional space makes welfare separation typical in the augmented family, with separation holding for almost every perturbation of each fixed base game.
\end{abstract}

\section{Intuition}\label{sec:intuition}

Consider a game with \(n\) players, each having an open strategy set
\(X\subseteq\R^d\), so the strategy-profile space is
\(X^n\subseteq\R^{nd}\). The first-order conditions define a map
\(F(u,\,\cdot\,):X^n\to\R^{nd}\), whose zero set contains the equilibria.
For pairs of distinct profiles, write
\[
\Conf_2(X^n):=\{(x,y)\in X^n\times X^n:x\neq y\}.
\]
For a scalar criterion \(A\) that is \(C^1\) in profiles, pairs of distinct
equilibria with equal \(A\)-values are zeros of
\[
(x,y)\longmapsto
\bigl(F(u,x),F(u,y),A(u,x)-A(u,y)\bigr),
\qquad
\Conf_2(X^n)\longrightarrow\R^{2nd+1}.
\]
If this map is transverse to \(\{0\}\), it has no zero. By the preimage
theorem, a nonempty transverse preimage is a submanifold whose dimension is
the domain dimension minus the codimension of the target submanifold in the
codomain. Here a nonempty zero set would therefore have dimension
\[
\dim\Conf_2(X^n)-\operatorname{codim}_{\R^{2nd+1}}\{0\}
=2nd-(2nd+1)=-1,
\]
which is impossible.

Hence, what we will show is that transversality is typical. Starting with utilitarian welfare, we establish this transversality for typical games
by applying Hunt, Sauer, and Yorke's parametric transversality
theorem~\cite[Theorem~1]{hunt1992prevalence}.

To the best of our knowledge, the no-welfare-ties result established here is
new. Related genericity and prevalence methods, however, have an established
history in economic theory.
For almost all finite normal-form games, equilibria are regular and persist
locally and smoothly under payoff perturbations~\cite{vandamme1981regular}; in
finite-state, finite-action discounted stochastic games, stationary Markov perfect
equilibria are generically purifiable by vanishing private payoff
shocks~\cite{doraszelski2010regular}.
For \(C^2\) games with at least two players and compact convex polyhedral
strategy sets, finiteness of the pure-equilibrium set is finitely
prevalent~\cite{anderson2001genericity}.
On an open dense set of \(C^2\) games with at least two players and simplex
strategy sets, equilibria at which every player chooses a nonvertex strategy are
Pareto-inefficient~\cite{dubey1986inefficiency}.
With observable preferences and a pure-equilibrium selection that is \(C^1\)
at zero types, the joint space of regular \(C^3\) two-player games on open
strategy sets with pure equilibria and \(C^3\) disposition families contains relatively prevalent
sets where some preference distortions improve actual payoffs and regular
payoff-monotonic selection from full-support initial type distributions does
not jointly converge to undistorted payoff maximization~\cite{heifetz2007maximize}.
The present note establishes a further typical property: distinct pure
equilibria have distinct welfare values, without requiring equilibrium
finiteness or compact strategy spaces.

\section{Setup}\label{sec:setup}

Fix integers \(n,d\geq 1\). Let \(X\subseteq\R^d\) be nonempty and open.
The profile space is \(X^n\subseteq\R^{nd}\), and the space of \(C^2\)
utility profiles is

\[
\U:=C^2(X^n,\R)^n,
\]

with the usual \(C^2\) topology.

Throughout, we use the standard Euclidean inner products to identify each
finite-dimensional Euclidean space with its dual.

For \(u=(u_1,\ldots,u_n)\in\U\), write
\(x=(x_1,\ldots,x_n)\in X^n\), where \(x_i\in X\) is player \(i\)'s
strategy and \(x_{-i}:=(x_j)_{j\neq i}\) is the profile of the other players'
strategies. The set of pure-strategy Nash equilibria is

\[
\NE(u)
:=
\left\{
x\in X^n:
u_i(x_i,x_{-i})\geq u_i(a_i,x_{-i})
\text{ for every }a_i\in X\text{ and every }i
\right\}.
\]

Define utilitarian welfare by

\[
W:\U\times X^n\to\R,
\qquad
W(u,x):=\sum_{i=1}^n u_i(x).
\]

Utilities are cardinal payoffs on specified scales; the welfare rule is held
fixed as payoff functions vary.

Define the first-order-condition map

\[
F:\U\times X^n\to(\R^d)^n\cong\R^{nd},
\qquad
F(u,x):=\bigl(D_i u_i(x)\bigr)_{i=1}^n,
\]

where \(D_i u_i(x)\) is the derivative of \(u_i\) with respect to player
\(i\)'s own strategy. Since \(X\) is open, every pure-strategy Nash
equilibrium is an interior optimum in each player's own strategy and hence

\begin{equation}\label{eq:NE-implies-FOC}
x\in\NE(u)\quad\Longrightarrow\quad F(u,x)=0.
\end{equation}

Define the equilibrium-candidate welfare-tie map

\begin{equation}\label{eq:tie-map}
T:\U\times\Conf_2(X^n)\to\R^{2nd+1},
\qquad
T(u,x,y)
:=
\bigl(F(u,x),F(u,y),W(u,x)-W(u,y)\bigr).
\end{equation}

For fixed \(u\), the zeros of \(T(u,\,\cdot\,,\,\cdot\,)\) are pairs of
distinct profiles that satisfy every player's first-order condition and
tie in welfare. Therefore \eqref{eq:NE-implies-FOC} implies

\begin{equation}\label{eq:tied-equilibria-subset}
\begin{split}
&\left\{(x,y)\in\Conf_2(X^n):
x,y\in\NE(u),\ W(u,x)=W(u,y)\right\}\\
&\qquad\subseteq
\left\{(x,y)\in\Conf_2(X^n):T(u,x,y)=0\right\}.
\end{split}
\end{equation}

Section~\ref{sec:no-welfare-ties} shows that, typically,
\(T(u,\,\cdot\,,\,\cdot\,)\) has no zero.
By \eqref{eq:tied-equilibria-subset}, no pair of distinct Nash equilibria
then ties in welfare.

\section{No welfare ties}\label{sec:no-welfare-ties}

\subsection{Hermite interpolation and the probe}

Let \(\Pcal_{\leq 3}(\R^{nd},\R^n)\) denote the vector space of polynomial
maps from \(\R^{nd}\) to \(\R^n\) of degree at most three. Its key property
is two-point Hermite interpolation: its elements can match arbitrarily
specified values and first derivatives at any two distinct points. The first
jet records exactly these data. For \(q\geq1\) and \(z\in\R^{nd}\), let

\[
J_z^1(\R^{nd},\R^q)
:=\R^q\times\operatorname{Lin}(\R^{nd},\R^q).
\]

If \(O\subseteq\R^{nd}\) is open, \(f\in C^1(O,\R^q)\), and \(z\in O\),
define

\[
j_z^1f:=\bigl(f(z),Df(z)\bigr)
\in J_z^1(\R^{nd},\R^q).
\]

For utility profiles, abbreviate
\(J_z^1:=J_z^1(\R^{nd},\R^n)\).

For distinct \(x,y\in\R^{nd}\), write
\(J_{x,y}^1:=J_x^1\times J_y^1\) and
\(j_{x,y}^1f:=\bigl(j_x^1f,j_y^1f\bigr)\) whenever \(f\) is defined on
an open set containing \(x\) and \(y\).

Let \(P\) be the restriction of
\(\Pcal_{\leq 3}(\R^{nd},\R^n)\) to \(X^n\).

\begin{lemma}[Two-point Hermite interpolation]\label{lem:hermite}
Fix distinct \(x,y\in X^n\). Then the two-point jet evaluation map

\[
j_{x,y}^1:
P\longrightarrow J_{x,y}^1
\]

is surjective and hence a submersion.
\end{lemma}

\begin{proof}
Because the interpolation problem is componentwise, it suffices to treat the
scalar-valued case. Fix an arbitrary element
\(\bigl((a,g),(b,h)\bigr)\in
J_x^1(\R^{nd},\R)\times J_y^1(\R^{nd},\R)\) of the codomain, where
\(a,b\in\R\) and \(g,h\in(\R^{nd})^\ast\). Choose an affine linear function
\(t:\R^{nd}\to\R\) such that \(t(x)=0\) and \(t(y)=1\). Define

\[
H_0(s):=2s^3-3s^2+1,
\qquad
H_1(s):=-2s^3+3s^2,
\]

and set

\[
q(z)
:=
aH_0(t(z))+bH_1(t(z))
+(1-t(z))^2g(z-x)+t(z)^2h(z-y).
\]

Then \(q\) has degree at most three and satisfies

\[
q(x)=a,
\quad Dq(x)=g,
\quad q(y)=b,
\quad Dq(y)=h.
\]

Applying this construction independently to each component and restricting
the resulting polynomial map to \(X^n\) proves the claim.
\end{proof}

We use the following probe formulation: a set \(A\subseteq\U\) is
\emph{finitely prevalent with probe
\(P\)} if it contains a Borel set \(B\subseteq\U\) such that, for every
\(u\in\U\),

\[
u+p\in B
\qquad\text{for Lebesgue almost every }p\in P.
\]

A set is \emph{generic} if it contains a countable intersection of open dense
sets, and \emph{typical} if it is both generic and finitely prevalent.

\subsection{Application to welfare ties}

We use the following corollary of Hunt, Sauer, and Yorke's parametric
transversality theorem \cite[Theorem~1]{hunt1992prevalence}, specialized to
the target submanifold \(\{0\}\) and applied in local coordinates.

\begin{corollary}[Parametric avoidance of zero]\label{cor:parametric-avoidance}
Let \(A\) and \(B\) be open subsets of finite-dimensional, second-countable
smooth manifolds without boundary, and let
\(\Phi:A\times B\to\R^q\) be \(C^1\), with \(q>\dim A\).
If \(\Phi\) is transverse to \(\{0\}\), then for almost every \(b\in B\),

\[
\Phi(a,b)\neq0\qquad\text{for every }a\in A.
\]

Here almost every means outside a set of Lebesgue measure zero in local
coordinates on \(B\).
\end{corollary}

\begin{theorem}[No utilitarian welfare ties]\label{thm:no-welfare-ties}
The set

\[
\left\{
u\in\U:
\text{for all }x,y\in\NE(u),\ x\neq y
\Longrightarrow
\sum_{i=1}^n u_i(x)\neq\sum_{i=1}^n u_i(y)
\right\}
\]

is typical in \(\U\), with probe \(P\).
\end{theorem}

\begin{proof}
Let
\[
\mathcal G
:=\{u\in\U:T(u,x,y)\neq0
\text{ for every }(x,y)\in\Conf_2(X^n)\}.
\]
Fix \(u\in\U\) and define the \(C^1\) map
\[
\Phi_u:\Conf_2(X^n)\times P\to\R^{2nd+1},
\qquad
\Phi_u(x,y,p):=T(u+p,x,y).
\]
For each \(x\neq y\), Lemma~\ref{lem:hermite} lets a perturbation
\(h\in P\) prescribe independently all utility values and first derivatives
at \(x\) and \(y\). In particular, it can prescribe both FOC blocks and
the welfare difference independently. Since \(T\) is linear in utilities,
\(D_p\Phi_u(x,y,p)\) is the surjective map \(h\mapsto T(h,x,y)\).
Thus \(\Phi_u\) is a submersion.

As \(\dim\Conf_2(X^n)=2nd<2nd+1\),
Corollary~\ref{cor:parametric-avoidance} gives \(u+p\in\mathcal G\) for
almost every \(p\in P\), for every \(u\in\U\). Such perturbations exist
arbitrarily close to zero; continuity of \(p\mapsto u+p\) therefore also
makes \(\mathcal G\) dense.

The open set \(\Conf_2(X^n)\) admits a countable compact exhaustion.
Nonvanishing of \(T\) on each compact set is open in \(u\), by uniform
control of payoff values and first derivatives on its coordinate projections.
Thus \(\mathcal G\) is a dense \(G_\delta\), hence generic and Borel.
The same almost-everywhere conclusion establishes its finite prevalence
with probe \(P\). By \eqref{eq:tied-equilibria-subset}, every game in
\(\mathcal G\) has no distinct welfare-tied Nash equilibria, proving the
claim.
\end{proof}

\begin{corollary}[Unique equilibrium maximizing utilitarian welfare]
\label{cor:unique-welfare-maximizer}
For every game in the typical set of
Theorem~\ref{thm:no-welfare-ties}, a pure-strategy Nash equilibrium
maximizing utilitarian welfare, if one exists, is unique.
\end{corollary}

\begin{proof}
Two distinct equilibria maximizing utilitarian welfare would have the same
utilitarian welfare, contrary to Theorem~\ref{thm:no-welfare-ties}.
\end{proof}

\begin{remark}[Minimal regularity for this argument]
The first-order-condition map \(F\) uses one derivative of each utility,
and parametric transversality requires
\(T(u,\,\cdot\,,\,\cdot\,)\) to be at least \(C^1\). Thus
\(u_i\in C^2\) is the weakest standard regularity under which this proof
applies.
\end{remark}

\begin{remark}[Scope]
The theorem as stated concerns pure equilibria on an open strategy set and
utilitarian welfare; it does not assert equilibrium existence or treat mixed
strategies. Section~\ref{sec:generalizations} records extensions of the
welfare criterion, probe, and utility space.
\end{remark}

\section{Generalizations}\label{sec:generalizations}

\subsection{Payoff aggregators}

\begin{proposition}[No ties for payoff aggregators]
\label{prop:payoff-aggregators}
Fix \(H\in C^1(X^n\times\R^n,\R)\) with
\(D_vH(z,v)\neq0\) for every \((z,v)\), and define
\[
A_H(u,z):=H\bigl(z,u(z)\bigr).
\]
The set of games whose distinct pure-strategy Nash equilibria have distinct
values of \(A_H\) is typical in \(\U\), with probe \(P\).
\end{proposition}

The gradient condition means that, holding strategies fixed, changing at
least one player's payoff changes the index to first order. The typical set
may depend on the fixed aggregator \(H\).

\begin{proof}
Replace \(T\) by
\[
T_H(u,x,y):=\bigl(F(u,x),F(u,y),A_H(u,x)-A_H(u,y)\bigr).
\]
For each fixed \(u\), the map \((x,y,p)\mapsto T_H(u+p,x,y)\) is \(C^1\).
Lemma~\ref{lem:hermite} lets a perturbation \(h\in P\) prescribe both FOC
blocks and the values \(h(x),h(y)\) independently. Setting \(h(y)=0\),
the payoff-value contribution to the differential of the last coordinate is
\[
D_vH\bigl(x,(u+p)(x)\bigr)h(x).
\]
This can take any scalar value without altering either FOC block. Hence the
parameter derivative is surjective. Corollary~\ref{cor:parametric-avoidance}
gives absence of zeros of \(T_H(u+p,\,\cdot\,,\,\cdot\,)\) for almost every
\(p\in P\), for every base game \(u\).

As in the proof of Theorem~\ref{thm:no-welfare-ties}, nonvanishing on each
compact subset of \(\Conf_2(X^n)\) is open in the \(C^2\) topology.
Compact exhaustion and arbitrarily small good perturbations therefore make
the set of games with no zeros of \(T_H\) a dense \(G_\delta\). It is the
Borel witness for finite prevalence, and its games have no distinct
\(A_H\)-tied Nash equilibria.
\end{proof}

\begin{corollary}[Weighted welfare]\label{cor:weighted-welfare}
For each fixed \(\lambda\in[0,\infty)^n\setminus\{0\}\), distinct
pure-strategy Nash equilibria typically have distinct values of
\(W_\lambda(u,x):=\sum_i\lambda_i u_i(x)\).
\end{corollary}

\begin{proof}
Take \(H(z,v)=\sum_i\lambda_i v_i\), so \(D_vH=\lambda\neq0\).
\end{proof}

\begin{corollary}[Unique equilibrium maximizing a payoff aggregator]
\label{cor:unique-aggregator-maximizer}
For each fixed \(H\) as in Proposition~\ref{prop:payoff-aggregators} and
every game in its typical set, a pure-strategy Nash equilibrium maximizing
\(A_H\), if one exists, is unique.
\end{corollary}

\begin{proof}
Two distinct maximizing equilibria would tie in \(A_H\), contrary to
Proposition~\ref{prop:payoff-aggregators}.
\end{proof}

Appendix~\ref{app:first-jet-criteria} extends the separation result to
scalar criteria that also depend on payoff derivatives.

\subsection{Alternative probes}

A finite-dimensional linear subspace \(P_0\subseteq\mathcal U\) is a
\emph{probe for \(\mathcal U\)} if, for every \(u\in\mathcal U\), the map
\((p,x,y)\mapsto j_{x,y}^1(u+p)\) is \(C^1\) on
\(P_0\times\Conf_2(X^n)\) and its derivative in \(p\) is surjective
everywhere. Addition uses the chosen linear structure on
\(\mathcal U\). Under the ordinary pointwise operations, joint \(C^1\)
regularity follows from the \(C^2\) regularity of the utilities, and

\[
j_{x,y}^1(u+p)=j_{x,y}^1u+j_{x,y}^1p,
\]

so the probe condition reduces to surjectivity of \(j_{x,y}^1|_{P_0}\).
By Lemma~\ref{lem:hermite}, the cubic space \(P\) is a probe for every
pointwise-linear utility space containing it. Any finite-dimensional
subspace containing a probe is also a probe in this pointwise setting.

\subsection{Admissible utility spaces}\label{sec:admissible-spaces}

\begin{definition}[Admissible utility space]
A class \(\mathcal U\subseteq C^2(X^n,\R)^n\) is \emph{admissible} if
it is endowed with a topological vector-space structure for which the
evaluation map

\[
\operatorname{ev}:\mathcal U\times X^n\longrightarrow\R^n,
\qquad
\operatorname{ev}(u,z):=u(z),
\]

is jointly continuous. The linear operations need not be pointwise.
\end{definition}

The space \(C^2(X^n,\R)^n\), with its usual \(C^2\) topology, is
admissible, as is every pointwise-linear subspace with the induced topology.
This includes \(\Pcal_{\leq k}(\R^{nd},\R^n)|_{X^n}\) for every \(k\geq3\).

The conclusions of Theorem~\ref{thm:no-welfare-ties} and
Proposition~\ref{prop:payoff-aggregators} hold typically in any admissible
utility space admitting a probe, for each fixed \(H\). The probe condition
and the proposition's rank argument give a \(C^1\) submersion along each
probe translate, hence almost-everywhere absence of ties.
Joint evaluation continuity and continuity
of \(H\) make the relation of \(A_H\)-tied Nash pairs closed, since every
Nash inequality and aggregator equality is a closed condition. Its projection
over each compact subset of \(\Conf_2(X^n)\) is closed, so compact exhaustion
makes the Nash no-ties set \(G_\delta\). Small good probe perturbations make
this set dense; it is also the Borel witness for finite prevalence.

\section{Arbitrarily small augmentation}\label{sec:augmentation}

Return to utilitarian welfare on the open strategy space \(X\).
Even when welfare separation is not typical in an admissible utility class,
adjoining arbitrarily small payoff perturbations from a suitable probe makes
separation generic in the augmented family and full measure along every
base-game fiber.

Let \(\mathcal V\) be an admissible utility space as in
Section~\ref{sec:admissible-spaces}, and let \(P_0\) be any probe for
\(C^2(X^n,\R)^n\) with its pointwise linear operations. Thus \(P_0\) is
a finite-dimensional pointwise-linear space of payoff perturbations with
\(j_{x,y}^1|_{P_0}\) surjective for every \(x\neq y\).
Form the external direct sum, with the product topology,
\[
\widetilde{\mathcal V}:=\mathcal V\oplus P_0=\mathcal V\times P_0,
\qquad
\widetilde{\operatorname{ev}}((u,p),z):=u(z)+p(z).
\]
The embedding \(u\mapsto(u,0)\) preserves the original linear structure.
Pairs are retained as parameters, even if they represent the same payoff
profile; genericity and prevalence refer to this parameter space.

\begin{corollary}[Small augmentation]\label{cor:small-augmentation}
Welfare separation is typical in \(\widetilde{\mathcal V}\), with probe
\(\{0\}\times P_0\). Fix any norm on \(P_0\). For each \(u\in\mathcal V\)
and \(\varepsilon>0\), a generic, full-Lebesgue-measure subset of
\(\{p\in P_0:\|p\|<\varepsilon\}\) yields no welfare ties between distinct
Nash equilibria. For every compact \(S\subset X^n\),
\[
\max_{|\alpha|\le2}\sup_{z\in S}\|D^\alpha p(z)\|
\le C_{S,P_0}\|p\|,
\]
with \(C_{S,P_0}>0\) independent of \(u\) and \(p\).
The bound also holds with \(S=X^n\) if all elements of \(P_0\) have
bounded derivatives of orders zero through two on \(X^n\).
\end{corollary}

\begin{proof}
Translation by \((0,q)\) adds the payoff \(q\), so the probe condition and
the proof of Theorem~\ref{thm:no-welfare-ties} give full measure along every
probe translate. Joint evaluation continuity makes the relation of
welfare-tied Nash pairs closed. Compact projection and exhaustion of
\(\Conf_2(X^n)\) therefore make the no-ties set \(G_\delta\); small good
probe perturbations make it dense. This also supplies the Borel witness for
finite prevalence. On each fixed-base probe ball, the good set is again
\(G_\delta\) and has full measure, hence is generic. The bound follows by
expanding \(p\) in a basis of \(P_0\) and bounding the basis functions and
their derivatives on \(S\), or on all of \(X^n\) under the stated boundedness
condition.
\end{proof}

For each fixed \(H\) as in Proposition~\ref{prop:payoff-aggregators},
Corollary~\ref{cor:small-augmentation} also holds with welfare replaced by
\(A_H\), evaluated at augmented payoffs as \(H(z,u(z)+p(z))\).
The proposition supplies the same rank argument along \(\{0\}\times P_0\),
and continuity of \(H\) preserves the closed relation of tied Nash pairs.

For example, the cubic polynomial space \(P\) is such a probe by
Lemma~\ref{lem:hermite}. When \(X\) is precompact, its polynomial derivatives
are bounded on \(\overline{X^n}\), giving uniform \(C^2\) smallness on the
whole domain.

\begin{remark}[A globally bounded probe]\label{rem:bounded-probe}
Set \(w(z):=e^{-\|z\|^2}\) and \(P_{\mathrm b}:=wP\).
At each \(z\), multiplication by \(w\) induces the invertible first-jet map
\[
(v,L)\longmapsto\bigl(w(z)v,\,w(z)L+v\,Dw(z)\bigr).
\]
Thus Lemma~\ref{lem:hermite} gives two-point first-jet surjectivity for
\(P_{\mathrm b}\). Its elements are restrictions of Gaussian-weighted
polynomials, whose derivatives of orders zero through two are bounded on
\(\R^{nd}\). Corollary~\ref{cor:small-augmentation} therefore gives uniform
\(C^2\)-smallness on the whole domain even when \(X\) is unbounded.
\end{remark}

Thus \(\mathcal V\times\{p\in P_0:\|p\|<\varepsilon\}\) is an
arbitrarily small augmentation with generic welfare separation, retaining
the original class at \(p=0\). Smallness holds through second derivatives
on compact subsets, and uniformly under the boundedness condition above.
Preserving additional economic restrictions requires perturbation directions
that respect them.

\appendix
\section{Criteria depending on payoff derivatives}
\label{app:first-jet-criteria}

In the full utility space \(\U=C^2(X^n,\R)^n\), with its usual \(C^2\)
topology, the same argument separates scalar criteria involving marginal effects.
For \(L\in\operatorname{Lin}(\R^{nd},\R^n)\), write \(L_{ij}\) for the
derivative block of output \(i\) with respect to input block \(j\).
When \(L=Du(z)\), these blocks are \(L_{ij}=D_j u_i(z)\).

\begin{theorem}[No ties for nonflat scalar first-jet criteria]
\label{thm:scalar-nonflatness}
Fix a \(C^1\) function
\[
\gamma:X^n\times\R^n\times\operatorname{Lin}(\R^{nd},\R^n)\to\R,
\qquad
A(u,z):=\gamma\bigl(z,u(z),Du(z)\bigr).
\]
Suppose that
\[
D_{(v,(L_{ij})_{i\neq j})}\gamma(z,v,L)\neq0
\qquad\text{for every }(z,v,L).
\]
Then the set of games whose distinct pure-strategy Nash equilibria have
distinct values of \(A\) is typical in \(\U\), with probe \(P\).
\end{theorem}

The condition requires a first-order response to payoff levels or cross
effects while holding the own-strategy derivatives fixed. For example, it
holds for \(A(u,z)=D_j u_i(z)\) when \(d=1\) and \(i\neq j\).
The typical set may depend on the fixed criterion \(A\).

\begin{proof}
Define
\[
T_A(u,x,y):=\bigl(F(u,x),F(u,y),A(u,x)-A(u,y)\bigr).
\]
For each fixed \(u\), the map \((x,y,p)\mapsto T_A(u+p,x,y)\) is \(C^1\).
Lemma~\ref{lem:hermite} lets a perturbation prescribe the diagonal
derivative blocks at \(x\) and \(y\), hence both FOC blocks. The nonflatness
condition lets its payoff values and cross derivatives at \(x\) supply any
remaining change to the last coordinate of the parameter differential,
without altering either FOC block or the jet at \(y\). The parameter
derivative is therefore surjective, and
Corollary~\ref{cor:parametric-avoidance} gives absence of zeros for almost
every \(p\in P\), for every base game \(u\).

As in the proof of Theorem~\ref{thm:no-welfare-ties}, nonvanishing on
compact configuration sets is open in the \(C^2\) topology. Compact
exhaustion and arbitrarily small good perturbations make the no-zero set
a dense \(G_\delta\) and a Borel witness for finite prevalence. Every game
in this set has no distinct \(A\)-tied Nash equilibria.
\end{proof}

\bibliographystyle{plain}
\bibliography{references}

\end{document}